%% file: Edgematching_arxiv.tex
\documentclass[12pt]{article}

\usepackage {graphicx}
\usepackage [latin1]{inputenc}
\usepackage {bm}
\usepackage {amsmath, amssymb}
\usepackage {listings}
\usepackage {tikz}
\usepackage {xspace}
\usepackage[a4paper,left=2cm,right=2cm,top=2cm,bottom=2.5cm]{geometry}
\usepackage{ifthen}
\usetikzlibrary{shapes,calc}
\input{tikzMakros}

\newtheorem{theorem}{Theorem}[section]

\newtheorem{problem}[theorem]{Problem}
\newenvironment{proof}[1][Proof]{\begin{trivlist}
\item[\hskip \labelsep {\bfseries #1}]}{\hfill q.e.d.\end{trivlist} }
\newtheorem{corollary}[theorem]{Corollary}

\newcommand{\oneinthreesat}{Monotone 1-in-3-SAT\xspace}
\newcommand{\ie}{i.\,e.\xspace}
\newcommand{\eg}{e.\,g.\xspace}

\newcommand{\refeq}[1]{(\ref{#1})}

\newcommand{\eq}[1]{``#1''}
\newcommand{\norm}[1]{\left\vert#1\right\vert}

\tikzstyle{whitecirc}=[circle,fill=white,draw=black,minimum size=24pt]

\newcommand{\colcirc}[5]{
\node(#5) [circle,minimum size=#3,fill=white,draw=black] at (#1,#2) {#4};}

\newcommand{\piece}[7]{\coordinate (LL) at (#1,#2); \coordinate (LR) at ($(#1,#2)+(#3,0)$);
\coordinate (UR) at ($(#1,#2)+(#3,#3)$); \coordinate (UL) at ($(#1,#2)+(0,#3)$);
\draw[auto] (LR) to node {#4} (UR);
\draw[auto] (UR) to node {#5} (UL);
\draw[auto] (UL) to node {#6} (LL);
\draw[auto] (LL) to node {#7} (LR);
\draw (LL) -- (UR); \draw (UL) -- (LR);
}

\newcommand{\piecename}[8]{\coordinate (LL) at (#1,#2); \coordinate (LR) at ($(#1,#2)+(#3,0)$);
\coordinate (UR) at ($(#1,#2)+(#3,#3)$); \coordinate (UL) at ($(#1,#2)+(0,#3)$);
\draw[auto] (LR) to node {#4} (UR);
\draw[auto] (UR) to node {#5} (UL);
\draw[auto] (UL) to node {#6} (LL);
\draw[auto] (LL) to node {#7} (LR);
\draw (LL) -- (UR); \draw (UL) -- (LR);
\node[whitecirc] at ($0.5*(LL)+0.5*(UR) $) {#8};
}

\newcommand{\brokenpiece}[4]{\coordinate (LL) at (#1,#2); \coordinate (LR) at ($(#1,#2)+(#3,0)$);
\coordinate (UR) at ($(#1,#2)+(#3,#3)$); \coordinate (UL) at ($(#1,#2)+(0,#3)$);
\draw[auto] (UL) -- (LL);
\draw[auto] (LL) -- ++(#3/3, 0) -- ++(#3/4, #3/4) -- ++(- #3/4, #3/4) -- ++(#3/4, #3/4) -- ++(- #3/4, #3/4) -- (UL);
\ifthenelse{\equal{#4}{true}}{
  \draw[auto] (UR) -- (LR);
  \draw[auto] (LR) -- ++(- #3/2, 0) -- ++(#3/4, #3/4) -- ++(- #3/4, #3/4) -- ++(#3/4, #3/4) -- ++(- #3/4, #3/4) -- (UR);
}{}
}

\newcommand{\piececolor}[7]{\coordinate (LL) at (#1,#2); \coordinate (LR) at ($(#1,#2)+(#3,0)$);
\coordinate (UR) at ($(#1,#2)+(#3,#3)$); \coordinate (UL) at ($(#1,#2)+(0,#3)$);
\coordinate (CE) at ($0.5*(LL)+0.5*(UR)$);
\draw[fill=#4] (LR)-- (CE)-- (UR) -- (LR);
\draw[fill=#5] (UR)-- (CE)-- (UL) -- (UR);
\draw[fill=#6] (UL)-- (CE)-- (LL) -- (UL);
\draw[fill=#7] (LL)-- (CE)-- (LR) -- (LL);
\draw (LL) -- (LR) -- (UR) --(UL) -- (LL);
\draw (LL) -- (UR); \draw (UL) -- (LR);
}

\newcommand{\settext}[3]{
\node at (#1,#2) {#3};}

\usepackage{amssymb}

\begin{document}

\title{ Edge-matching Problems with Rotations\footnote{A preliminary version
of this paper appeared in Proceedings of
FCT 2011,~\cite{DBLP:conf/fct/EbbesenFW11}.}}

\author{Martin Ebbesen\\Appstract Consulting\\ DK-1300 K\o benhavn K \\ \texttt{martin@appstract.dk}
\and
Paul Fischer\\DTU Compute\\ Technical University of Denmark\\ DK-2800 Lyngby\\ \texttt{pafi@dtu.dk}
\and
Carsten Witt\\DTU Compute\\ Technical University of Denmark\\ DK-2800 Lyngby\\ \texttt{cawi@dtu.dk}}
\maketitle

\begin{abstract}
Edge-matching problems, also called edge matching puzzles,  are abstractions  of placement problems
with neighborhood conditions. Pieces with colored edges have to be placed
on a board such that adjacent edges have the same color.
The problem has gained interest recently with the (now terminated) Eternity~II puzzle,
and new complexity results.
In this paper we consider a number of settings which differ in size of the puzzles
and the manipulations allowed on the pieces. We investigate the effect of allowing
rotations of the pieces on the complexity of the problem, an aspect
that is only marginally treated so far.
We show that some problems have
polynomial time algorithms while others are NP-complete.
Especially we show that allowing rotations in one-row puzzles
makes the problem NP-hard. The proofs of the hardness result
uses a large number of colors.  This is essential because we also show that this problem (and another
related one) is
fixed-parameter tractable\footnote{An NP-complete problem is fixed-parameter tractabel
if there is algorithms which is exponential in only one parameter which specifies the
problem size an polynomial in the size of the input.
When this parameter is constant (or ``small''), then the problem is efficiently
solvable.}, where the relevant parameter is the number of colors.
\end{abstract}

\textbf{Keywords:}
Edge-matching puzzles, complexity, fixed-parameter tractability

\section{Introduction}
We consider the following combinatorial optimization problem. Let $N$ and $M$
be positive integers.
The puzzles considered in this paper consist of $NM$ quadratic \emph{pieces} whose
edges are colored. Let $c_0,c_1,\ldots,c_K$ denote the colors.
The pieces are placed in the cells of a rectangular $N\times M$ grid.
The edges of a piece are denoted
\emph{right}, \emph{top}, \emph{left}, \emph{bottom}  in the obvious way.
A piece $P$ can then be specified by its position $(i,j)$ on the grid
and the colors $c_{r,P},c_{t,P},c_{l,P},c_{b,P}$ on its
right, top, left, and bottom edge, in this order.  The image below shows a
piece specified by the color pattern  $(c_1,c_2,c_3,c_4)$.
\begin{center}
    \begin{tikzpicture}
       \piece{1}{1}{1.5}{$c_1$}{$c_2$}{$c_3$}{$c_4$}
    \end{tikzpicture}

\end{center}

Two pieces are \emph{neighbors} if they are placed in the same column
and the difference
of their row  coordinates in the grid is $1$ (or with the roles of
rows and columns interchanged), i.e., they share an edge.
Given two pieces $P$ and $P'$ on neighboring positions, say at $(i,j)$
and $(i,j+1)$ respectively, we say that they \emph{match}, if the adjacent edges
have the same color, that is $c_{r,P}=c_{l,P'}$. In this case we also say that
the edge they have in common \emph{matches}.
Similarly, if they are
vertically adjacent, say at $(i,j)$ and $(i+1,j)$, they \emph{match}, if
$c_{b,P}=c_{t,P'}$, see example below.

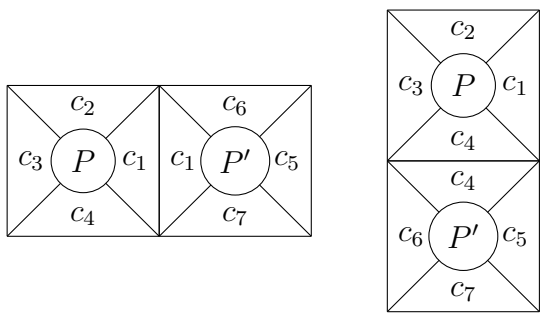
\begin{center}
\begin{tikzpicture}
\piecename{1}{2}{2}{$c_1$}{$c_2$}{$c_3$}{$c_4$}{$P$}
\piecename{3}{2}{2}{$c_5$}{$c_6$}{$c_1$}{$c_7$}{$P'$}

\piecename{6}{3}{2}{$c_1$}{$c_2$}{$c_3$}{$c_4$}{$P$}
\piecename{6}{1}{2}{$c_5$}{$c_4$}{$c_6$}{$c_7$}{$P'$}
\end{tikzpicture}
\end{center}

We say that a board is \emph{solved} if the pieces are placed in such a way that
all edges match. At this point one has to specify rules for the \eq{border edges},
that is, the edges facing the border of the board. We consider three cases here
\begin{description}
  \item[Free] There is no restriction, any color matches the border of the board.
  \item[Monochrome] There is a single color for all border edges.
  \item[Cyclic] We assume that the board actually is a torus, that is, the
  top edges are aligned with the bottom ones and the left edges are aligned
  with the right ones.
\end{description}
An arrangement of the pieces which solves the board is called a \emph{solution}.

The manipulations allowed in a board are:
\begin{description}
  \item[Swap] Two pieces $P$ and $P'$ interchange their position, without
  being rotated.
  \item[Rotate] A piece $P$ is rotated counter-clockwise
  in-place by $0\deg$, $90\deg$, $180\deg$, or $270\deg$.
  \item[Edge permutation] The colors on the edges are permuted in one of the
  $24$ possible ways. The position of the piece is unchanged. Note that this
  manipulation includes rotations.
  \item[Flip] The colors of the left and right or up and down edges of piece $P$
  are interchanged (that is $P = (c_1,c_2,c_3,c_4)$ becomes $P_{flipped} = (c_3,c_2,c_1,c_4)$ or $P = (c_1,c_2,c_3,c_4)$
  becomes $P_{flipped} = (c_1,c_4,c_3,c_2)$).
  Flipping is an in-place operation.
\end{description}
Combinations of the manipulations are possible.

The \emph{edge matching problem} is then formulated as follows.
\begin{problem}
Given is a $N\times M$ board, $K$ colors, $NM$ pieces, a border rule and a set of
manipulations. The decision problem is given by the question ``is it possible to solve the board?''

In the optimization version of the problem, a solution has to be produced.
\end{problem}

Edge-matching has found  applications in biology
where it is used in a method for DNA fragment assembly~\cite{ptw}.
The problem has also gained interest recently with the Eternity II puzzle, a $16\times 16$-puzzle,
with a \$2 million prize for a solution (which was not found, though).
Another area where this kind of problems
appears is  chip-layout, where interfaces have to be placed on a rectangular chip but their order
is arbitrary.

\subsection{Previous Work}

Edge-matching with one row, swaps and no rotation corresponds to domino tiling \cite{LesniakOrtrud} and is,
as we discuss in this paper (and as covered in \cite{LesniakOrtrud}),
equivalent to the problem of finding an Eulerian path in a multi-graph.
The computational complexity of edge-matching and related problems has been studied
for several decades. Early results show that the more general tiling problem, \ie,
solving the edge-matching problem
with swaps for a quadratic board
using only a subset of a given set of
un-rotatable pieces is NP-complete \cite{GareyJohnson}.
Goles and Rapaport showed in~\cite{gr97} that the edge-matching problem with only
in-place rotations (disallowing swaps)
and free border rule is NP-complete.
More recently, Demaine and Demaine \cite{DemaineJigsaw} proved that the edge-matching variant
considered here
 is  NP-complete for quadratic boards with swaps and rotations. In 2010,
Antoniadis and Lingas showed that this problem is even APX-hard, \ie, hard
to approximate, already for boards with at least two rows \cite{AntoniadisL10}.
\subsection{Overview of the Paper}

Most of the above-mentioned analyses  of the edge-matching problem
consider swaps, but
do not allow rotating pieces. The NP-hardness proof in \cite{DemaineJigsaw}, even though
it formally allowed
rotations, forces the pieces to be used in a fixed orientation.
Only
recently, the APX-hardness proof in \cite{AntoniadisL10} explicitly made use of
rotation and swaps at the same time.
In this paper we address the question of what changes
in the complexity of the problem occur when rotations of pieces are allowed.

In Section~\ref{sec:one-row} we consider puzzles with only one row ($1\times M$).
For these puzzles there is a known correspondence to Euler paths, which we
describe along with some previous results.
We then show that single-row puzzles where only in-place rotations are allowed
can be efficiently solved. In contrast we show that solving single-row
puzzles with rotations and swaps allowed is an NP-hard problem. The proofs
implicitly use the Euler path formulation of the problem.

In Section~\ref{sec:multirow} we strengthen a result of~\cite{DemaineJigsaw}  by showing
that already boards with two rows with swaps only are NP-hard to solve.

Having observed that the hardness proofs use a number of colors which is proportional
to the number of pieces in the puzzle, we show in Section~\ref{sec:constColor} that
the problem of solving a single-row puzzle with rotations and swaps is efficiently
solvable when the number of colors is constant or logarithmic in the number of pieces.
Thus this problem is fixed-parameter tractable. Similarly, an efficient solution
is also possible for the case
of constant board size and number of colors.

\section{Boards with One Row \label{sec:one-row}}
Edge-matching with one row ($N=1$) without rotation is equivalent to
the problem of finding an Eulerian trail in a multi-graph allowing
loops. An Eulerian trail is a trail that visits each edge of the
graph exactly once, and the concept is applicable to both directed
and undirected graphs: An undirected graph is Eulerian (i.e.
contains an Eulerian trail) iff it is connected and has no more than
2 vertices with an uneven number of edges. An Eulerian circuit (a
trail starting and ending on the same vertex) requires that all
vertices must have an even number of edges. A directed graph has an
Eulerian circuit iff it is weakly connected and all vertices has
equal in- and out-degree, while a path requires connectedness and 0
or 2 vertices with difference in in- and out-degree equal to~1. The
above facts are shown in, \eg, \cite{LesniakOrtrud}.

\subsubsection*{Swaps and flipping}
An edge-matching instance with flipping, with $M$ pieces and $K$
colors is transformed to an undirected multigraph containing $K$ vertices
and $M$ edges. Each piece corresponds to an edge connecting the vertices
corresponding to the colors on opposing edges of the piece.
Now traversing a vertex using two different
edges corresponds to matching two pieces having a common color. A
trail in the graph corresponds to a matched chain of pieces, and an
Eulerian trail corresponds to a solution where all pieces are
fitted, and vice versa.  Hence this variant of the problem is efficiently
solvable.
Figure~\ref{swapflip:fig} shows an example, see \cite{EbbesenMaster} for more details.
\begin{figure}[h]
\begin{center}
\begin{tikzpicture}
\piece{0}{2}{1.4}{1}{0}{2}{0}
\piece{0}{0}{1.4}{1}{0}{3}{0}
\piece{2}{0}{1.4}{4}{0}{3}{0}
\piece{2}{2}{1.4}{4}{0}{3}{0}

\colcirc{5}{3}{1cm}{1}{C1}
\colcirc{7.5}{3}{1cm}{3}{C3}
\colcirc{7.5}{0.5}{1cm}{2}{C2}
\colcirc{5}{0.5}{1cm}{4}{C4}

\draw[thick] (C1) -- (C2);
\draw[thick] (C1) -- (C3);
\draw[thick] (C4) to [bend right = 10](C3);
\draw[thick] (C4) to [bend left = 10] (C3);

  \piece{9}{2}{1.2}{4}{0}{3}{0}
  \piece{10.2}{2}{1.2}{3}{0}{4}{0}
  \piece{11.4}{2}{1.2}{1}{0}{3}{0}
  \piece{12.6}{2}{1.2}{2}{0}{1}{0}

    \end{tikzpicture}
\end{center}
\caption{An example for a single-row puzzle with swaps and flips. The colors
of the vertical edges are denoted $1$, $2$, $3$, and $4$; the irrelevant color
at the top and bottom is $0$. The multi-graph has a one node for each relevant
color and one edge for every piece. The edge connects the nodes corresponding
to the colors of the vertical edges. The graph has an Euler trail, $2,1,3,4,3$
and thus the puzzle is solvable. The corresponding solution is shown at right.}
\label{swapflip:fig}
\end{figure}
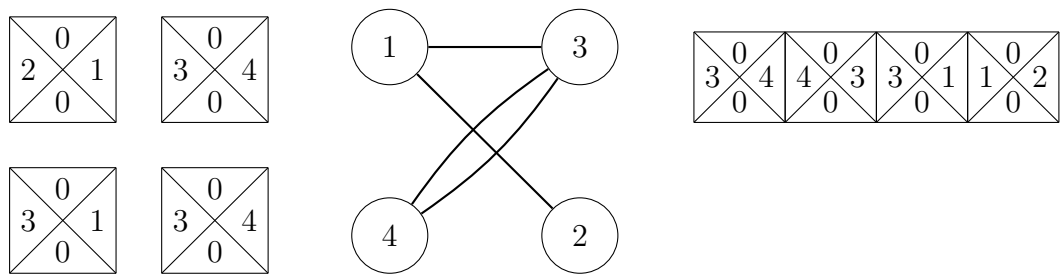

\subsubsection*{Swaps only}
Edge-matching with swaps (no flipping) correspond to Eulerian cycles in directed multigraphs, and vice versa.
Hence this variant of the problem is efficiently solvable.

\subsubsection*{Rotations only}\label{oneDRot:sec}
With free border rule, if the board consists of only one piece it will
 always be solved.
 Obviously, for a board with two pieces there will be a solution if the
 two pieces can be rotated such that their touching edges match.
 This can be generalized:

\begin{theorem}
Single-row  edge matching puzzles with in-place rotations can be solved or determined to
be unsolvable by an algorithm that has time complexity linear in the number
of pieces.
\end{theorem}
\begin{proof}
The proof is by induction. The pieces are numbered from
left to right as $p_{0},p_{1},...,p_{M-1}$. $L(p_{i})$ is the set
of colors that can be on the left edge of piece $p_{i}$ which is
equivalent to the set of unique colors on the 4 edges of $p_{i}$.
$R(p_{i})$ is the set of colors that can be on the right edge of
$p_{i}$ such that $p_{i}$ fits $p_{i-1}$, that is, such that $L(p_{i})\cap R(p_{i-1})\neq \emptyset$.
$L(p_{i})$ and $R(p_{i})$ can both have at most 4 members since
a piece has 4 edges.

\emph{Base case}: For $p_{0}$ there is no constraint, as one piece
always represents a solution to a $1\times 1$ board: $R(p_{0})=L(p_{0})$.

\emph{Induction step}: Given that $R(p_{i})$ is known and the board
is solvable up to $p_{i}$ then piece $p_{i+1}$ can be fitted if
$R(p_{i})\cap L(p_{i+1})\neq \emptyset$. If not, the board must
be unsolvable. $R(p_{i+1})$ can be calculated by finding the color
on the opposite edge of $p_{i+1}$ for each member in $R(p_{i})\cap L(p_{i+1})$.
This operation takes constant time because of the bound on the sizes
of the two sets.

The theorem holds because each step in the induction, \ie, each additional
piece, only adds a constant time overhead.
\end{proof}

\subsubsection*{Swaps and Rotations}
\begin{theorem}\label{swaprot:thm}
If both swaps and rotations are allowed, single-row edge-matching
with free border rule is NP-complete.
\end{theorem}

\begin{proof}
We use a polynomial-time reduction from the
NP-complete \oneinthreesat problem \cite{Schaefer:1978:CSP:800133.804350}:
given $m$ monotone
(\ie, disallowing negation) clauses
$c_1,\dots,c_m$ over $n$ variables $x_1,\dots,x_n$, the question is whether there
exists an assignment which satisfies exactly one variable in each clause.

Let $m_i$, $i=1,\dots,n$, denote the number of occurrences of variable $i$, \ie,
$m_1+\dots+m_n=3m$. The instance of \oneinthreesat is mapped to a single-row
board with free border at the top and bottom and cyclic border to the left and right (which can
be simulated by a free border using a polynomial number of extra pieces).
There are  $M=10m$ pieces and $K=13m-n+1$ colors. We denote colors by lower-case and pieces
by upper-case letters. The set of pieces
 consists of all $V_{j,q}$, $V'_{j,q}$, $S_j$,
where $j=1,\dots,m$ and $q=1,2,3$, and $A_{i,k}$, where $i=1,\dots,n$ and
$k=1,\dots,m_i$. The colors are $\ell$, $t_{j,q}$,
$t'_{j,q}$,  $f_{j,q}$, $s_j$, and $a_{i,k}$, where $k=1,\dots,m_i-1$. We define and name the different classes
of pieces as follows:
\begin{description}
\item[Value] $V_{j,q}:=(f_{j,q}, t_{j,q}, s_j, \ell)$; $V_{j,q}'=(\ell, t'_{j,q}, f_{j,q}, \ell)$.
\item[Satisfying] $S_j:=(s_j,s_j,s_j,s_j)$.
\item[Accordance] If $m_i=1$ then $A_{i,1}:=(\ell,\ell,\ell,\ell)$. Otherwise,
let the $k$\nobreakdash-th occurrence of~$x_i$ be at position~$q_k\in\{1,\dots,3\}$
in clause $c_{j_k}$,
$k=1,\dots,m_i$. Then
define $A_{i,k}: = (t_{j_k,q_k}', a_{i,k-1}, t_{j_k,q_k}, a_{i,k} )$, where
$a_{i,0}:=a_{i,m_{i}}:=\ell$.
\end{description}
This completes the transformation, which is obviously polynomial-time computable.

Rotations
can deactivate up to two colors from a piece, namely those that are facing the border
and thus need not be matched. Active colors must occur an even number of times
in a solved board.
When placing a piece $V_{j,q}$, either
$f_{j,q}$ and $s_j$, or
$t_{j,q}$ and $\ell$ will be deactivated.
Actually, since $f_{j,q}$ only appears in $V_{j,q}$ and $V_{j,q}'$, a solution
requires that these pieces either both deactivate $f_{j,q}$ or that both
activate this color. If they both activate $f_{j,q}$ (which also activates~$s_j$ once), they will be adjacent
in a solution. Later, active $f_{j,q}$ will
model that the $q$\nobreakdash-th variable in~$c_j$ is false, and inactive $f_{j,q}$
will model a true setting. Since $s_j$ appears only in $V_{j,1},V_{j,2},V_{j,3}$ and~$S_j$, but cannot be deactivated
in~$S_j$, where it will be used twice, it is necessary for a solution that an even, positive number of
pieces from $V_{j,1}, V_{j,2}, V_{j,3}$
use~$s_j$.
%
This forces
two of the variables in~$c_j$
to a \eq{false} and one to a \eq{true} setting.

Finally, the interplay of the $t_{j,q}$ and $t_{j,q}'$ within the $V$-
and $A$\nobreakdash-pieces will ensure the consistency of the truth assignment. We fix a variable~$x_i$ and
assume~$m_i>1$ as
there is nothing to show otherwise.
Since color $a_{i,k}$, $k=1,\dots,m_{i}-1$,
is only used in $A_{i,k}$ and $A_{i,k+1}$, these two pieces must either both be rotated to activate
the $t$ and~$t'$ colors on their edges, or the pieces must be adjacent,
which inductively requires that the solution contains all $A_i$\nobreakdash-pieces as a
chain $A_{i,1},A_{i,2},\dots,A_{i,m_i}$
in this or reverse order. In the case of a chain, the colors $t_{j_{k},q_k}$
and $t'_{j_{k},q_k}$ must be inactive for all $k=1,\dots,m_i$, which means
that the pieces $V_{j_{k},q_k}$ and $V'_{j_{k},q_k}$ are adjacent by virtue of color
$f_{j_{k},q_k}$, hence a \eq{false} setting of~$x_i$. Otherwise, the colors
are all active, which is (by the uniqueness of the colors) only possible if $A_{i,k}$ is fit
between the pieces $V_{j_{k},q_k}$ and $V'_{j_{k},q_k}$, both of which are rotated
to a \eq{true} setting. Hence, all $m_i$ occurrences of~$x_i$ must be consistent.

We now prove that there is a solution to \oneinthreesat if and only if there is one
to the puzzle. Let us first prove the implication $\Rightarrow$, \ie,  we consider an assignment
to~$x_1,\dots,x_n$ which satisfies exactly one variable in each clause. We lay down the pieces
as subsequences in clause-wise order. When considering a clause $c_j=(u_{j,1}\vee u_{j,2} \vee u_{j,3}  )$, $j=1,\dots,m$, let $q\in\{1,\dots,3\}$ be the index
of its unique satisfied variable, say this variable is~$x_i$ in its $k$\nobreakdash-th occurence.
The subsequence for~$c_j$ starts with
$V_{j,q}$ rotated by 270\textdegree{} (left-hand side has color~$\ell$), followed by
$A_{i,k}$ unrotated and $V'_{j,q}$ rotated by 90\textdegree{}. Note that the
right-hand side is~$\ell$. Let $q'$ and $q''$ be the indices of the other two
unsatisfied variables in~$c_j$. We proceed by placing $V'_{j,q'}$
rotated by 180\textdegree{}, then $V_{j,q'}$ rotated by 180\textdegree{} and
afterwards $S_j$. The construction for clause $j$ is continued by placing
$V_{j,q''}$ and $V'_{j,q''}$ unrotated, which again ends in color~$\ell$. Note
that we have placed all $V_{j,\cdot}$\nobreakdash-pieces and the $S_{j}$\nobreakdash-piece as well as
 a single $A$\nobreakdash-piece corresponding
to the occurrence of the satisfied variable.
Finally, if they have not been placed before,
we use all $A$\nobreakdash-pieces
corresponding to the unsatisfied variables as follows: If $u_{j,q'}=x_r$
then $A_{r,1},A_{r,2},\dots,A_{r,m_r}$ are appended in the chained
way described above (again
ending in~$\ell$), completed by the chain for variable $u_{j,q''}$.
The construction ends with~$\ell$ on the right-hand side and is continued
with the next clause, resulting in all pieces being used and the puzzle
being solved.

For the implication $\Leftarrow$, assume now that the puzzle is solved.
We have already argued that for any~$j$, there must be exactly
one~$q$ such that the pair $(V_{j,q},V'_{j,q})$
is rotated according a \eq{true} setting
and two other~$q$ such that the pair is in a
\eq{false} setting. We set the variables in~$c_j$ accordingly.
If a variable contains in another clause, we already know that
the corresponding pieces must be rotated in a consistent way,
which proves that we have a solution to \oneinthreesat.
\end{proof}
\section{Boards with at Least Two Rows}\label{sec:multirow}
We consider boards with two rows and arbitrarily many columns. Of course,
the roles of columns and rows can be interchanged.

In \cite{DemaineJigsaw} edge-matching is shown to be NP-complete for quadratic
boards. More recently, \cite{AntoniadisL10} showed in a much more involved proof
that the problem is even APX-complete, already for rectangular
boards with only two rows. We focus on Demaine's and Demaine's technique from
\cite{DemaineJigsaw}, which does not use rotations, and
show that it can be strengthened to also include
boards with row-count two -- by extension edge-matching with any width-height ratio is NP-complete
(with the exception of single-row puzzles),
since it is trivial to force a board to contain more rows by adding uniformly colored pieces.
However,  as stated in Section~\ref{sec:one-row}, edge-matching without rotation is efficiently
solvable when the row-count is 1.

\begin{theorem}
Two-row  edge matching puzzles with swaps and free or monochromatic border is NP-complete.
\end{theorem}
\begin{proof}
The transformation is from 3-partition in which the task is to partition a set of $3m$
positive integers into $m$ sets,
each consisting of $3$ integers such that the integers in each set sum to the same value $S$.
This problem can be visualized as the task of, given a collection of bars of varying length,
placing the bars in rows, such that each row has $3$ bars and all rows have the width $S$.
Importantly 3-partition is also NP-complete when all the integers are limited to values in
the range $(S/4,S/2)$,
meaning that any row with width S must contain exactly $3$ bars~--~this
is the version of the problem used in this proof. The problem remains NP-complete
when $S$ and $m$ are polynomially related.

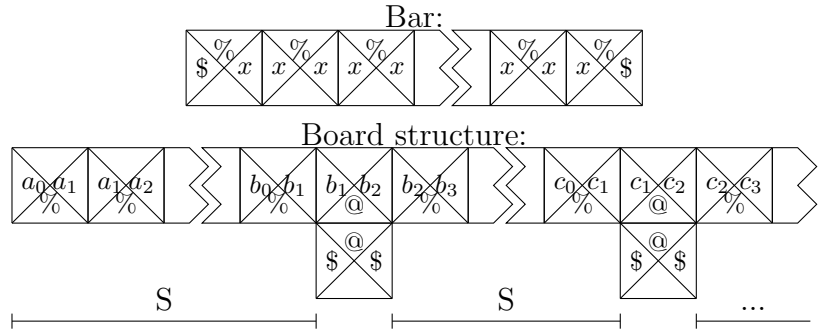
\begin{center}

Bar:

\begin{tikzpicture}
font=\footnotesize

\piece{1}{1}{1}{$x$}{$\%$}{$\$$}{$ $}
\piece{2}{1}{1}{$x$}{$\%$}{$x$}{$ $}
\piece{3}{1}{1}{$x$}{$\%$}{$x$}{$ $}
\brokenpiece{4}{1}{1}{true}
\piece{5}{1}{1}{$x$}{$\%$}{$x$}{$ $}
\piece{6}{1}{1}{$\$$}{$\%$}{$x$}{$ $}
\end{tikzpicture}

Board structure:

\begin{tikzpicture}
font=\footnotesize

\piece{1}{1}{1}{$a_1$}{$ $}{$a_0$}{$\%$}
\piece{2}{1}{1}{$a_2$}{$ $}{$a_1$}{$\%$}
\brokenpiece{3}{1}{1}{true}
\piece{4}{1}{1}{$b_1$}{$ $}{$b_0$}{$\%$}
\piece{5}{1}{1}{$b_2$}{$ $}{$b_1$}{$@$}  \piece{5}{0}{1}{$\$$}{$@$}{$\$$}{$ $}
\piece{6}{1}{1}{$b_3$}{$ $}{$b_2$}{$\%$}
\brokenpiece{7}{1}{1}{true}
\piece{8}{1}{1}{$c_1$}{$ $}{$c_0$}{$\%$}
\piece{9}{1}{1}{$c_2$}{$ $}{$c_1$}{$@$}  \piece{9}{0}{1}{$\$$}{$@$}{$\$$}{$ $}
\piece{10}{1}{1}{$c_3$}{$ $}{$c_2$}{$\%$}
\brokenpiece{11}{1}{1}{false}

font=\normalsize
\draw[auto] (1, -0.2) -- (1, -0.4); \draw[auto] (5, -0.2) -- (5, -0.4); \draw[auto] (1, -0.3) to node{S} (5, -0.3);
\draw[auto] (6, -0.2) -- (6, -0.4); \draw[auto] (9, -0.2) -- (9, -0.4); \draw[auto] (6, -0.3) to node{S} (9, -0.3);
\draw[auto] (10, -0.2) -- (10, -0.4); \draw[auto] (10, -0.3) to node{...} (11.5, -0.3);
\end{tikzpicture}

\end{center}

Converting this problem into an edge-matching puzzle with height $2$ proceeds as follows:
A section of pieces (a 'bar') is defined for each of the $3m$ integers.
The internal left-right edge-color (shown as 'x' in the figure) is unique for each bar,
and every bar starts and ends on the specific color '\$'.
A board with height $N=2$ and width $M=mS+m-1$  is constructed,
where the upper row is forced to have a particular layout as illustrated in the figure
by using unique colors for every edge pair,
and the lower row is separated into areas of length $S$, each of which can contain $3$ bars.
All separators will fit any bar left and right (color '\$'),
and all bars will fit in any position below the fixed upper row (color '\%').
If the 3-partition has a solution then it will be possible to place the $m$ bars
into the board giving a solution to the edge matching puzzle.
If the edge matching puzzle can be solved it will be because all sections can be
placed into the forced layout of the grid,
meaning that there is a solution to the corresponding 3-partition problem.
\end{proof}

\input{constColors_arxiv.tex}

\section{Summary and  Conclusions}
In this paper, we have focused on several problem aspects of edge-color puzzles
whose impact on complexity was unknown or only marginally treated
before, in particular
rotating pieces, in-place pieces, border rules, number of colors
and shapes
of pieces. With regard to single-row boards, it
has been shown that introducing rotations makes the
otherwise easy problem NP-hard. In basically all hardness results,
a large number of colors is used, usually linear in the number of pieces.
For two cases we have shown that the problem is fixed-parameter tractable,
i.e., it becomes efficiently solvable when the number
of colors is constant or logarithmic in the number of pieces.

The table below summarizes the known hardness results for boards with one row
compared to boards with at least two rows.

\begin{center}
\begin{tabular}{|c||c|c|}
\hline
Rows: & 1 & $\geq2$\tabularnewline
\hline
\hline
Swap & P & NP-complete\tabularnewline
\hline
Rotation & P & NP-complete\tabularnewline
\hline
Both & NP-complete, fixed-parameter tractable & NP-complete\tabularnewline
\hline
\end{tabular}
\end{center}

This paper has raised further questions
concerning the hardness
of edge-matching problems. For example, are there more hard problems
which are fixed-parameter tractable?  Furthermore,
it is unknown whether the single-row case with rotations is
hard to approximate.

\section*{Acknowledgement}

The second author gratefully acknowledges support by the
DTU's Corrit travel grant.




  \bibliographystyle{elsarticle-num}{}
  \bibliography{edgematching}

%



\end{document}

%% file: tikzMakros.tex
\usepgflibrary{arrows}

\tikzstyle{vertex}=[circle,fill=black!25,minimum size=24pt,inner sep=0pt]
\tikzstyle{node}=[circle,thick,draw=blue!75,minimum size=24pt]
\tikzstyle{nodeblack}=[circle,thick,draw=black,minimum size=24pt]
\tikzstyle{nudenode}=[]
\tikzstyle{nodegray}=[circle,thick,draw=blue!75,fill=black!10,minimum size=24pt]
\tikzstyle{diredge}=[line width=1pt]
\tikzstyle{edge}=[line width=1pt]
\tikzstyle{edgeb}=[line width=2pt]
\tikzstyle{edgebb}=[line width=4pt]
\tikzstyle{bendedge}=[line width=1pt,bend left=15]

\tikzstyle{roundRect}=[rounded corners, line width = 1pt]
\tikzstyle{roundNode}=[rectangle,thick,rounded corners,draw=blue!75,minimum size=24pt]

\newcommand{\edgebblbl}[3]{\draw [-,edgebb] (#1)to node[auto] {#3} (#2);}

%% file: constColors_arxiv.tex
\section{Dependence on the number of colors.}\label{sec:constColor}
In the proof of Theorem~\ref{swaprot:thm} (and other hardness results concerning  puzzles),
the number of colors used is linear in the
the size of the reference problem. It is natural to ask whether the problem
remains hard, when the number of colors is limited to a constant.
We now show that the problem of  Theorem~\ref{swaprot:thm} becomes efficiently solvable if the number of colors
is at most $c(\log(M))^{1/4}$, for $c$ constant.

\begin{theorem}
If both swaps and rotations are allowed, single-row edge-matching
with free border and two colors is solvable in time $O(M)$, where
$M$ is the number of pieces. For $K$ colors the problem is solvable
in time $O(M+5^{K^4})$.
\end{theorem}
The last part of the theorem leads to the following corollary.
\begin{corollary}
If both swaps and rotations are allowed, the single-row edge-matching
problem with free border is fixed-parameter tractable.
\end{corollary}

\begin{proof} (Of the theorem)
We start by exemplifying the proof method by looking at
the case of one and two colors, black ($b$) and white ($w$).

If there is only one color, the puzzle is solvable, by placing
the pieces in a row.

With exactly two colors, the pieces can be classified as either
\emph{constants} or \emph
{switches}. Constants are pieces
that have the same color on the left and right edge, regardless of the rotation.
Switches can be rotated such that
there are different colors on the left and right edge.

\begin{center}
\begin{tikzpicture}
\piececolor{0}{0}{1.5}{white}{white}{white}{white}
\piececolor{2}{0}{1.5}{black}{black}{black}{black}
\piececolor{4}{0}{1.5}{white}{black}{white}{black}

\piececolor{7}{0}{1.5}{black}{white}{white}{white}
\piececolor{9}{0}{1.5}{white}{black}{black}{black}
\piececolor{11}{0}{1.5}{white}{white}{black}{black}

\settext{3}{-0.8}{Constants}
\settext{10}{-0.8}{Switches}
\end{tikzpicture}
\end{center}

If the puzzle consists of constants only it can be solved if and only if
it does not contain both a completely black and a completely white piece.

Otherwise there is both a monochromatic black piece and a white one.
We claim that the puzzle has a solution if and only if there is at least
one switch. Consider the multi-graph modeling as defined in Section~\ref{sec:one-row}. 
Recall that an Euler path corresponds to a solution when swaps and flips
are allowed, but not general rotations. We now show, how this method can 
be used also in the case of few colors when also rotations are allow.
For two colors the graph
has two nodes $w$ and $b$, one for each color.
Consider the constant pieces first. The monochromatic ones give rise
to loops at the node corresponding to their color.
The non-monochromatic pieces give rise to a loop on one of the nodes.
Which node this is depends on the rotation. In absence of a switch piece
there is no edge between the nodes, hence no Euler path and no solution.

If the puzzle contains at least one switch, then this can be rotated in
such a way that it forms an edge $\{b,w\}$ between the two nodes. Regardless
how the other pieces are rotated, the resulting graph has an Euler path.
One starts at, say, $w$ and traverses all loops at that node, then one uses
one edge $\{w,b\}$ to get to node $b$. Then one traverses all loops at
$b$. Finally one uses the remaining $\{w,b\}$-edges (if any) exactly once.
The Euler path corresponds to a solution as described above.

\newcommand{\makenodesize}[5]{\node (#5) at (#1,#2)  [circle,draw=black,minimum size=#3] {#4}; }

Consider the case of a puzzle with $M$ pieces and $K$ colors.
Let $X = abcd$ be a \emph{color scheme}, i.e.,
a counter-clockwise order of four (not necessarily distinct)
colors. We assume that the cyclic order of the color scheme
is such that it is lexicographically minimal.
This is to avoid treating an
$abcd$ and $bcda$ as different color schemes.

Every arrangement of the
$M$ pieces gives rise to a multi-graph.
The graph has $K$ nodes and one edge for every
piece. A piece with color scheme $abcd$ will correspond to an edge $\{a,c\}$ if $a$
and $c$ are the colors on the vertical edges. Otherwise the piece gives rise
to an edge $\{b,d\}$. Let $m_{a,b}$ denote the number of edges between nodes $a$ and $b$.
Loops ($a=b$) are allowed, however they do not have any influence in the solvability
of the puzzle, as discussed in the example for two colors.

Given an arrangement of the  $M$  pieces,
let $n_{X}$ be the number of pieces having color scheme $X=abcd$.
Let $n_{X,ac}$ denote the number of pieces with color scheme $X$
which are rotated to become an $\{a,c\}$-edge. Similar $n_{X,bd}$.
We make the convention that $n_{X,ac} = 0$ if $X$ does not contain
the colors $a$ and $c$ on opposing edges. Then
\begin{equation}
n_{X} = n_{X,ac} + n_{X,bd} \label{constraint1:eqn}
\end{equation}
and
\begin{equation}
m_{ac} = \sum_X n_{X,ac} \label{constraint2:eqn}
\end{equation}
Furthermore
\begin{displaymath}
deg(a) = \sum_{b\colon b > a} m_{ab} + \sum_{c\colon c < a}m_{ca} + m_{aa}
\end{displaymath}

For the graph to have an Euler cycle the $m_{xy}$ have to be chosen
such that they obey the constraints~(\ref{constraint1:eqn}) and~(\ref{constraint2:eqn})
and  all degrees are even.
For the graph to have an Euler path the $m_{xy}$ have to be chosen
such that they obey the constraints~(\ref{constraint1:eqn}) and~(\ref{constraint2:eqn})
and exactly two nodes have odd degrees.

As we are only interested in connectivity and parity of degrees,
we replace the $n_{X}$ as follows (that is we remove pieces with
color scheme $X$):
\begin{itemize}
  \item If $n_{X} \le 4$ it is unchanged.
  \item If $n_{X} > 4$ and even then set $n_{X} = 4$.
  \item If $n_{X} > 4$ and odd then set $n_{X} = 3$.
\end{itemize}

We claim that this does not change the fact whether the puzzle is solvable
or not.
Assume  $n_{X} > 4$ and is even. The number of pieces rotated such that 
they produce an $\{a,c\}$-edge is $n_{X,ac}$, number of pieces rotated such that
they produce an $\{b,d\}$-edge is $n_{X,bd}$ and $n_{X} = n_{X,ac} + n_{X,bd}$.

Assume that the puzzle is solvable. Then there is a way to choose $n_{X,ab}$ and
$n_{X,bd}$
such that the resulting graph has an
Euler cycle or  path, especially it is connected.
We only discuss the case of an Euler cycle, the arguments for
a path are identical.
Then the degrees of all nodes, especially $a,b,c,d$
are even. As $n_{X}$ is reduced by an even number but remains at least $3$,
we can reduce both $ n_{X,ac}$ and $n_{X,bd}$ (that is remove the
corresponding edges) by an even number each,
without being forced to make one of them zero.
If some $ n_{X,ac}$ is zero in the beginning, it will remain unchanged, therefore
we only consider $ n_{X,ac} > 0$ in the following.
The resulting graph is connected and
all degrees are even. Hence it still has an Euler cycle.

Now assume that the puzzle is not solvable. Then there is no way to
choose the numbers $n_{X,ab}$ and $n_{X,bd}$
such that the resulting graph has an
Euler cycle. That is the graph is un-connected or has nodes
with odd degree for any such choice. If there is no way to choose
$n_{X,ab}$ and $n_{X,bd}$ to make the
graph connected, then reducing the number of pieces maintains that property.

Assume therefore that for some choices of the numbers $n_{X,ab}$ and $n_{X,bd}$
the graph is connected, but always has nodes with odd degree.

Consider a given partition of the $n_{X} = n_{X,ac} + n_{X,bd}$.
If $n_{X}$ is odd then one of  $n_{X,ac}$ and  $n_{X,bd}$ is odd
the other is even. Replacing $n_{X}$ by $3$ (and changing
$n_{X,ac}$ and $n_{X,bd}$ to satisfy~\refeq{constraint1:eqn}) maintains this property.
Thus it is not possible to introduce a new distribution of parities
of the nodes (for example one with only even node degrees).

If $n_{X}$ is even then $n_{X,ac}$ and  $n_{X,bd}$ are either
both even of both odd. Replacing $n_{X}$ by $4$
and changing
$n_{X,ac}$ and $n_{X,bd}$ to satisfy~\refeq{constraint1:eqn}) maintains this property.
If $n_{X,ac}$ and  $n_{X,bd}$ were both even and positive,
 we set both to $2$. If both were odd and positive, we set
them to $3$ and $1$. Again it is not possible to introduce a new distribution of parities
of the nodes (for example one with only even node degrees).

Hence reducing the $n_{X}$ as described still allows the
graph to be connected and does not allow the graph to have only vertices of even degree.

Now, check all possibilities to partition the reduced $n_{X}$ into
$n_{X,ac}$, $n_{X,bd}$ and compute the $m_{xy}$. For each partition check
whether it leads to a degree sequence ensuring an Euler cycle.

Each $n_{X}$ can be partitioned into $n_{X,ac}$ and  $n_{X,bd}$
in $4$, or $5$ ways, depending on it being $3$
or $4$; these are the possibilities 
for the latter case $(0,4)$, $(1,3)$, $(2,2)$, $(3,1)$, and $(4,0)$.

Hence the number of possibilities to try is at most $5^{\ell}$,
where $\ell$ is the number of color schemes. The quantity
$\ell$ can be upper-bounded by $K^4$, where $K$ is the number of colors.
For each partition it has to be checked, whether the resulting degree
sequence is Eulerian.
Computing the $n_{X}$ is linear in the number of pieces.

Hence the total time for deciding a puzzle with
$M$ pieces and $K$ colors is
\[
O(M+5^{K^4})
\]
which is linear for constant $K$, and polynomial for $K=c(\log(M))^{1/4}$, for $c$ constant.
\end{proof}

The last result shows that the number of colors seems to play an essential role for complexity 
of puzzle problems. The following result supports this by showing that also puzzles 
with in-place rotations can be efficiently solved when the number of colors and rows is 
constant. As mentioned in section the unrestricted problem is NP-complete, see ~\cite{gr97}.
This resembles fixed-parameter tractability, but requires two parameters to be constant.

\begin{theorem}
Puzzles with $N$ rows, $M$ columns, and $K$ colors
and in-place rotations are solvable in
polynomial time for constant $N$ and $K$.
\end{theorem}

\begin{proof} 
Consider a column in the puzzle. For each of the $N$ pieces in the column,
there are 4 possible rotations, giving $4^N$ configurations for the column.
An \emph{$N$-color pattern} is a sequence $[c_1,\ldots,c_N]$ of $N$ colors.
Given $K$ colors there are at most $K^N$ $N$-color patterns, which we number
$1,\ldots,K^N$.
For each column there are at most $4^N$ possible color patterns at the left
and at the right edge of the column.
Consider column $j$ and construct a bipartite graph $G_j$ as follows.
The vertex set is $V = V_{j,\mathrm{left}}\cup V_{j, \mathrm{right}}$
where $V_{\mathrm{left}}$ and $V_{\mathrm{right}}$ each contain $K^N$ vertices, one for
each color pattern.
Let $V_{j,\mathrm{left}}=\{v_{j,1},\ldots,v_{j,K^N}\}$
and  $V_{j,\mathrm{right}}=\{w_{j,1},\ldots,w_{j,K^N}\}$.
We then try all $4^N$ configurations of the pieces in column~$j$.
If a configuration gives color pattern $a$ at left and pattern $b$ at right
and the colors at the $N-1$ horizontal edges match
then we  add a directed edge $(v_{j,a},w_{j,b})$ to the graph. Let $G_j$ denote
the resulting bipartite graph.
The construction of $G_j$ can be performed in time $O((4^N)+(K^N))$ given that
the index of the color pattern can be computed in time $\Theta(K+N)$.
Graph $G_j$ has $2K^N$ nodes and at most $4^N$ edges.

Next we connect the consecutive graphs $G_j$, $G_{j+1}$, $j=1,\ldots,M-1$ by
adding the directed edges $(w_{j,a},v_{j+1,a})$, $a=1\ldots,K^N$. Finally we add
a source node $s$ and connect it to all left nodes of $G_1$ by directed edges
$(s,v_{1,a})$. We also add a sink node $t$ and connect
it to all right nodes of $G_M$ by directed edges
$(w_{M,a},t)$. Let $G=(V,E)$ be the resulting graph. Note that
$\norm{V} = 2MK^N + 2$ and $\norm{E} \le M4^N+(M-1)K^N+2K^N$
and that $G$ can be constructed in time $O(M(K^M+4^N))$.

Now, compute whether there is directed path from $s$ to $t$ in $G$.
We claim, that if so, there is solution for the puzzle otherwise there is not.
Assume there is a directed path from $s$ to $t$.
Let $(s,v_{1,a_1})$ be the first edge of the path and $(w_{n,b_n},t)$ be the last one.
In between there are alternating edges inside the $G_j$ and in between the $G_j$.
This means that for every $j=1,\ldots,n-1$ the path contains
an edge $(w_{j,c_j},v_{j+1,c_j})$ for some color pattern $c_j$. Hence there
are matching color patterns between all rows.
Within every $G_j$ the path uses an edge $(v_{j,a_j},w_{j,b_j})$. That is,
the  color patterns at the left  and at the right of column $j$ can be
realized simultaneously by an appropriate configuration of the pieces in that column.
Hence the puzzle has a solution. As finding the shortest path can be done
in time linear in the graph size, the total running time is polynomial in 
for constant  $N$ and $K$.

Conversely, any solution of the puzzle gives rise to at least one path from $s$ to
$t$ by the definition of the graph.

\end{proof}

%% file: Edgematching_arxiv.bbl
\begin{thebibliography}{1}
\expandafter\ifx\csname url\endcsname\relax
  \def\url#1{\texttt{#1}}\fi
\expandafter\ifx\csname urlprefix\endcsname\relax\def\urlprefix{URL }\fi
\expandafter\ifx\csname href\endcsname\relax
  \def\href#1#2{#2} \def\path#1{#1}\fi

\bibitem{DBLP:conf/fct/EbbesenFW11}
M.~Ebbesen, P.~Fischer, C.~Witt, Edge-matching problems with rotations, in:
  Fundamentals of Computation Theory - 18th International Symposium, {FCT}
  2011, Oslo, Norway, August 22-25, 2011. Proceedings, 2011, pp. 114--125.

\bibitem{ptw}
P.~Pevzner, H.~Tang, M.~Waterman, An {E}ulerian path approach to {DNA} fragment
  assembly, Proceedings of the National Academy of Sciences of the United
  States of America 98, no. 17 (2001) 9748--9753.

\bibitem{LesniakOrtrud}
O.~R. Lesniak, Linda;~Oellermann, {An Eulerian exposition}, Journal of Graph
  Theory 10 (1986) 277--297.

\bibitem{GareyJohnson}
M.~Garey, D.~S. Johnson, Computers and Intractability: A Guide to the Theory of
  NP-Completeness, W.H. Freeman and Company, 1979.

\bibitem{gr97}
E.~Goles, I.~Rapaport, Complexity of tile rotation problems, Theoretical
  Computer Science 188 (1997) 129--159.

\bibitem{DemaineJigsaw}
E.~D. Demaine, M.~L. Demaine, Jigsaw puzzles, edge matching, and polyomino
  packing: Connections and complexity, Graphs and Combinatorics 23 (2007)
  195--208.

\bibitem{AntoniadisL10}
A.~Antoniadis, A.~Lingas, Approximability of edge matching puzzles, in: Proc.\
  of SOFSEM~2010, Vol. 5901 of LNCS, Springer, 2010, pp. 153--164.

\bibitem{EbbesenMaster}
M.~Ebbesen, Analysis of restricted edge-matching problems, Master's thesis,
  Technical University of Denmark, master thesis, Technical University of
  Denmark, reference no. IMM-M.Sc.-2011-08,
  \texttt{http://etd.dtu.dk/thesis/275156/} (2011).

\bibitem{Schaefer:1978:CSP:800133.804350}
T.~J. Schaefer, The complexity of satisfiability problems, in: Proceedings of
  the tenth annual ACM symposium on Theory of computing, STOC '78, 1978, pp.
  216--226.

\end{thebibliography}
